\documentclass[10pt,a4paper]{article}
\def\oggi{4.12.12}
\usepackage[english]{babel}
\usepackage[T1]{fontenc}
\usepackage[dvips]{graphicx}
\usepackage{amsmath,amsfonts,amsthm}
\usepackage{url}
\usepackage{cite}

\numberwithin{equation}{section}

\newtheorem{hyp}{Assumption}
\newtheorem{teo}{Theorem}[section]
\newtheorem{defi}[teo]{Definition}
\newtheorem{definition}[teo]{Definition}
\newtheorem{proposition}[teo]{Proposition}
\newtheorem{lemma}[teo]{Lemma}

\newtheorem{corollary}[teo]{Corollary}
\newtheorem{remark}[teo]{Remark}

\def\B{{\mathcal B}}
\def\H{{\mathcal H}}
\def\I{{\mathcal I}}

\def\L{{\mathcal L}}
\def\Y{{\mathcal Y}}
\def\U{{\mathcal U}}
\def\Q{{\mathcal Q}}
\def\phib{{\bf \zeta}}

\font\strana=cmti10
\def\lie{\hbox{\strana \char'44}}

\def\bF{\textbf{F}}

\def\bx{\textbf{x}}
\def\bn{\textbf{n}}

\def\by{\textbf{y}}
\def\be{\textbf{e}}

\def\bY{\textbf{Y}}

\def\length#1{\ell(#1)}

\def\bL{\textbf{L}}

\def\M{\mathcal{M}}
\def\G{\mathcal{G}}

\def\N{\mathcal{ND}}

\DeclareMathOperator{\dist}{dist}

\title{Asymptotic behavior of an elastic satellite\\with internal friction}
\author{Emanuele Haus, Dario Bambusi}

\begin{document}

\date{\oggi}
\maketitle

\begin{abstract}
We study the dynamics of an elastic body whose shape and position
evolve due to the gravitational forces exerted by a pointlike
planet. The main result is that, if all the deformations of the satellite
dissipate some energy, then under a suitable nondegeneracy condition
there are only three possible outcomes for the dynamics: (i) the orbit
of the satellite is unbounded, (ii) the satellite falls on the planet,
(iii) the satellite is captured in synchronous resonance i.e. its
orbit is asymptotic to a motion in which the barycenter moves on a
circular orbit, and the satellite moves rigidly, always showing the
same face to the planet. The result is obtained by making use of
LaSalle's invariance principle and by a careful kinematic analysis
showing that energy stops dissipating only on synchronous orbits. We
also use in quite an extensive way the fact that conservative
elastodynamics is a Hamiltonian system invariant under the action of
the rotation group. 
\end{abstract}

\section{Introduction}

In this paper, we study the dynamics of an elastic satellite
interacting with a pointlike planet. Precisely, we study the dynamics
of an elastic body, moving in the gravitational field generated by a
pointlike mass.  We consider the equations of motion of continuum
mechanics, with body forces due to the gravitational fields and
internal traction arising from the body deformation, without
introducing any further approximation. We prove that, if the internal
structure of the satellite is such that any deformation dissipates
some energy and if a suitable nondegeneracy condition is satisfied,
then the dynamics of the system has only three possible final
behaviors:
\begin{itemize}
\item[(i)] the orbit of the satellite is unbounded;
\item[(ii)] the satellite falls on the planet;
\item[(iii)] the satellite is captured in synchronous resonance.
\end{itemize}
By item (iii) we mean that the shape of the body reaches a final
configuration, that its center of mass moves on a circular orbit and
that it always shows the same face to the planet, i.e. the planet is
at rest in a frame comoving with the satellite.

Concerning the inner structure of the body we make as few assumptions
as possible. Precisely we assume that the stress tensor is the sum of
two terms, the first one being Hamiltonian, i.e. it is the $L^2$
gradient of a ``stored energy functional'', and a second one being
nonconservative. On the second term we only assume that, as a
consequence of its presence, there is dissipation of energy at any
time at which the time derivative of the Cauchy Green stress tensor does not
vanish. 

The idea of the proof is to use LaSalle's principle (see
\cite{lasalle}), which is a generalization of Lyapunov
theorem. LaSalle's principle ensures that any precompact orbit
approaches an invariant set which is contained in the manifold where
the Lie derivative of the energy vanishes. The core of the paper
consists in characterizing such an invariant set. Since in such an
invariant set the dynamics is conservative, it turns out that a
convenient framework for our study is that of Hamiltonian systems with
symmetry as developed for example in \cite{marsdenhughes} or, in a
form directly useful for our problem, in \cite{SPM91}.

So we start by writing down the Lagrangian and the Hamiltonian of the
conservative part of the system and then we add to the equations of
motions the nonconservative forces.

Then we start analyzing the nondissipating manifold $\N$. We first
prove that $\N$ consists of rigid motions, and then we show that the
motions laying on $\N$ are actually circular orbits. Finally we show
that they are relative equilibria of the reduced Hamiltonian system
obtained by exploiting the rotational invariance of the original
Hamiltonian. At this point the application of LaSalle's principle
would allow to conclude that the orbit is asymptotic to a manifold
obtained by taking the union of all the synchronous orbits. In order
to prove that the system is actually asymptotic to a single
synchronous orbit we exploit the conservation of angular momentum and
we assume a nondegeneracy property stating that the relative
equilibria are isolated. This property is discussed in detail in
Section \ref{assu.deg} and we show that it is typically fulfilled.

\vskip15pt

The present result still has some quite strong limitations. The main
one is that we do not discuss existence and uniqueness for the Cauchy
problem of the equations we study. As is well known this is one of the
main open problems of elasticity, so we do not enter its
discussion. Here we limit ourselves to assuming that the system we study is well
posed. Our assumptions are true for example for any Galerkin cutoff of
the system or if one can show that the system behaves as a parabolic
system (as is expected in the dissipative case under consideration). 

The second limitation of our result rests in the fact that we assume
that the system is described by differential equations. This means
that we do not consider the case where the system is described by an
integrodifferential equation with delay, a case which can occur in
elasticity. The case we have in mind is the one in which the
dissipation is of the kind of that appearing in Navier Stokes
equation. We expect that our theory can be extended to the case with
delay, but for sure the methods should be adapted.

\vskip15pt

The study of the gravitational interaction between a deformable body
and a pointlike mass traces its origin back to the pioneering work by
Darwin \cite{darwin1,darwin2}. His work shows that, in some
approximation, the effect of the internal dynamics of the satellite is
just that of producing an effective dissipating effective force on the
orbital and spin degrees of freedom of the satellite. Darwin's work
was subsequently generalized by Kaula \cite{kaula} and many other
authors (for instance,
\cite{goldreich,macdonald,peale,alexander}). Critical reviews of the
work by Darwin, Kaula and followers can be found in
\cite{ferrazmello,efroimsky1,efroimsky2}. However, the Darwin-Kaula
procedure is heuristic and, from a mathematical point of view, its
range of validity is far from being clear. For this reason, in the
present paper (as in \cite{bambusihaus}), the point of view is that of
starting from first principles in order to obtain a rigorous
mathematical proof of the phenomena under consideration.

\vskip10pt We remark that the result of the present paper rules out
the possibility that periodic orbits different from synchronous
resonance exist. This is quite surprising, since some celestial bodies
are known to be in a spin-orbit resonances different from the
synchronous one (for example Mercury). Our interpretation is that
orbits like the 3:2 resonance of Mercury are very likely to be either
metastable or orbits eventually impacting the planet. We plan to
investigating further in this direction.

\vskip10pt

In Section \ref{setting} we state the main result of the present
paper: to this end, we recall the Lagrangian formalism for
elastodynamics, we write down the related Cauchy problem and we
formulate the nondegeneracy assumption. Section \ref{sec.3} is devoted
to the proof of the main result: we recall the statement of LaSalle's
invariance principle, prove that the only solutions which dissipate
no energy are synchronous orbits and apply La Salle principle to our
system. Finally, in Section \ref{assu.deg} we discuss the
nondegeneracy assumption and we prove that typically it is fulfilled.

\section{Statement of the main result}
\label{setting}

\subsection{The setting}
We study the dynamical system consisting of:
\begin{itemize}
\item[(i)] a pointlike mass $M$ (which we will sometimes call
  ``planet''), which is at rest and which is chosen as the origin of a
  system of coordinates; 
\item[(ii)] an elastic body, free to move in space; we will call this
  extended body ``satellite''.
\end{itemize}

To deal with elastodynamics we use the framework of \cite{marsdenhughes} and
\cite{SPM91} from which we take some notations and formalism, that we
now recall.

We denote by $\B\subset\Re^3$ the reference configuration of an
elastic body and assume that $\B$ is open and bounded with a smooth
boundary $\partial \B$. We define the configuration space $\Q$ to be a
Banach space of maps\footnote{Actually we should restrict ourselves to
  the manifold of the maps s.t.  $\det{D\zeta}>0$, however we will
  consider this as a condition on the domain of definition of the
  system.}  $\phib:\B\to\Re^3$.
Typically in elastodynamics one assumes $\Q\subset H^{s}(\B)$ with $s$
large enough; we will come back later to this point, for the moment we
simply assume that $\phib$ admits as many derivatives as needed.  

In the conservative case, classical three dimensional elasticity is a
Lagrangian system, the Lagrangian $\L:T\Q\to\Re$ is the
difference of kinetic and potential energy. In our case there are also
some dissipative forces that will be added to the Lagrange
equations. As usual $T\Q\simeq\Q\oplus\Q$ is the tangent bundle to
$\Q$.

We start by writing down the conservative part of the system. 
The Lagrangian $\L$ of the system is defined by
\begin{equation}
\label{lag}
\L=K-U_g-U_{sg}-U_e\ ,
\end{equation}
where 
\begin{align}
\label{kin}
K(\dot\zeta):=\frac{1}{2}\int_{\B}\rho_0(\bx)
\left|\dot\zeta(\bx)\right|^2d^3\bx
\\
\label{ug}
U_g(\zeta):=\int_{\B}\rho_0(\bx){V_g(\zeta(\bx))d^3\bx}\ ,
\\
\label{usg}
U_{sg}(\zeta):=\int_{\B}\rho_0(\bx){V^{\zeta}_{sg}(\zeta(\bx))d^3\bx}\ ,
\\
\label{ue}
U_e(D\zeta):=\int_{\B}{W(\bx, D\zeta(\bx))d^3\bx}\ ,
\end{align}
the functions $V_g$, $V_{sg}^{\zeta}$ are defined by
\begin{align}
\label{vg}
V_g(\chi):=-\frac{kM}{|\chi|}\ ,
\\
\label{vsg}
V_{sg}^{\zeta}(\chi):=-\int_{\B}{\frac{k\rho_0(\bx)}{|\zeta(\bx)-\chi|}d^3\bx}\ ,
\end{align}
and $W$ is the stored energy
function; $k$ is the universal gravitational constant, and $\rho_0\in
C^\infty(\B)$ the density of the body in the reference
configuration\footnote{Of course we assume that $\rho_0(\bx)\not=0$
  $\forall \bx\in\B$.}. The stored energy function is assumed to
depend on $\zeta$ only through the deformation gradient $\bF:=
D\zeta\equiv\{\partial \zeta^i/\partial x^a\}$. We assume that $W$ is
frame independent in the sense that
\begin{equation}
\label{2.7}
W(\bx,\bF)=W(\bx, R\bF)\ \text{for\ all}\ R\in SO(3)\ .  
\end{equation}

As shown in \cite{SPM91} this implies that the Kirchoff stress tensor,
namely 
\begin{equation}
\label{tauij}
\tau^i_j=\frac{\partial \zeta^i}{\partial
x^a}\frac{\partial W}{\partial(\partial \zeta^j/\partial
x^a)}\qquad \ (\text{sum\ over}\ a\ \text{understood})\ ,
\end{equation}
is symmetric.

Assuming the stress free boundary condition, namely
\begin{equation}
\label{stress.free}
\left.\frac{\partial W}{\partial(\partial \zeta^j/\partial
x^a)} n_a \right|_{\partial \B}=0\ ,
\end{equation}
where $\bn\equiv(n_1,n_2,n_3)$ is the external normal to $\partial\B$,
one deduces the standard Lagrange equations:
\begin{align}
\label{lag.eq}
\rho_0\ddot \zeta &=-\nabla_{\zeta}\L\equiv -\rho_0\frac{\partial
  V_g}{\partial \chi}(\zeta)-\rho_0\frac{\partial
  V_{sg}^{\zeta}}{\partial \chi}(\zeta)+ \frac{\partial}{\partial
x^a}\frac{\partial W}{\partial(\partial \zeta/\partial
x^a)}\ ,
\end{align}
where $\nabla_{\zeta}\L\equiv
(\nabla_{\zeta^1}\L,\nabla_{\zeta^2}\L,\nabla_{\zeta^3}\L) $ is as
usual the gradient with respect to the $L^2$ scalar
product\footnote{i.e. it is defined by $d\L(\zeta) h=\langle\nabla
  \L(\zeta);h\rangle_{L^2}$ for all $h\in\Q$.} and
coincides with the expression at r.h.s. of \eqref{lag.eq}.

\begin{remark}
\label{rot.inv1}
It is easy to check that, if a function $U(\zeta)$ is rotation
invariant, i.e. $U(R\zeta)=U(\zeta)$, $\forall R\in SO(3)$, then
\begin{equation}
\label{rot.pot}
[\nabla U](R\zeta)= R\nabla U(\zeta)\ .
\end{equation}
All the terms of the Lagrangian have this property.
\end{remark}

Since the Lagrangian is independent of time, the energy 
\begin{equation}
\label{ham}
\H=K+U_g+U_{sg}+U_e\ ,
\end{equation}
is formally conserved for the system (\ref{lag.eq}) with the
boundary conditions \eqref{stress.free}.
 
Furthermore, since the Lagrangian is invariant under the group action 
\begin{equation}
\label{group}
\Q\times SO(3)\ni(\zeta,R)\mapsto R\zeta\ ,
\end{equation}
by N\"other's theorem the quantity
\begin{equation}
\label{L}
\bL:=\int_{\B}\zeta(\bx)\times
\rho_0(\bx)\dot\zeta(\bx) d^3\bx\ , 
\end{equation}
 is conserved for the system. 
 Of course $\bL$ coincides with the total angular momentum. 

In order to get the equations governing the non conservative dynamics
one has simply to add the nonconservative forces\footnote{by this
  notation we mean that $G$ is a function of the functions $\zeta$,
  $\zeta_t$, not of their value $\zeta(\bx)$, $\zeta_t(\bx)$, so it
  can also depend on an arbitrary number of derivatives of such
  functions.}, $G=G(\zeta,\zeta_t)$ i.e. to substitute equation
\eqref{lag.eq} with the equation
\begin{equation}
\label{non.lag}
\rho_0\ddot \zeta =-\nabla_{\zeta}\L-G\ .
\end{equation}

In order to write down the precise assumptions on $G$
(which will be given in the next subsection) we have also to introduce
the (right) Cauchy Green deformation tensor $C:=(D\zeta)^TD\zeta$ or,
componentwise
\begin{equation}
\label{Cauchy}
C^{ij}=\sum_a\frac{\partial \zeta^i}{\partial x^a}\frac{\partial
  \zeta^j}{\partial x^a}\ . 
\end{equation}
As it is well known $C$ is symmetric, positive definite and allows to
write $D\zeta$ in the polar decomposition form
$D\zeta(\bx)=R(\bx)\sqrt{C(\bx)}$, where $R(\bx)\in SO(3)$ is a
rotation matrix.

\subsection{The Cauchy problem}
\label{sec:cauchy}

In the following we will always denote by
$y\equiv(\zeta,\dot \zeta)\in T\Q\simeq
\Q\oplus \Q$ a point in the space of
initial data for our system (which we keep distinct from the phase
space, in which the velocities will be substituted by the momenta). 

The problem of existence of solutions for the system \eqref{lag.eq},
\eqref{non.lag} with the boundary conditions \eqref{stress.free} is a
major problem of elastodynamics (see e.g. \cite{ball}). Here, we do not
want to enter such problem, so \textit{we will limit ourselves to
  assuming the needed well-posedness properties}. Furthermore we will
only consider classical solutions, so we give the following definition:

\begin{definition}
\label{sol}
Given $y_0\in T\Q$, a positive $T$ and a function $y\in
C^2((0,T);T\Q)$ we say that it is a solution of the system
\eqref{non.lag}, \eqref{stress.free} with initial
datum $y_0$, if it fulfills the equations and the boundary conditions
for all $t\in(0,T)$ and one has
$$
\lim_{t\to0^+}y(t)=y_0\ .
$$
\end{definition}

Then we also need the following definitions:
\begin{defi}\label{ni}
A solution is said to be impacting (in the future) if
$\inf_{t>0}\dist(\zeta(\B,t),0)=0$.
\end{defi}
\begin{defi}\label{ns}
A configuration is said to be non singular if $\det [D\zeta]>0$. 
\end{defi}

\begin{definition}
\label{reg.sol} 
A solution $y(t)$, $t\in(0,T)$ is said to be regular if it is non
impacting and the corresponding configuration is non singular for all
times.  An initial datum $y_0\equiv(\zeta_0,\dot \zeta_0)$ is regular
if $\dist(\zeta_0(\B),0)>0$ and $\det [D\zeta_0]>0$.
\end{definition}

\begin{definition}
\label{precomp}
A regular solution $y(t)$, $t\in(0,\infty)$ is said to be precompact
if, for any increasing sequence $\{t_n\}\subset(0,\infty) $ there
exists a subsequence $\{t_{n_k}\} $ s.t. the limit
$\lim_{k\to+\infty}y(t_{n_k})$ exists and the limit is
regular.\end{definition}

\begin{hyp}\label{hyp:cauchy.problem}
We assume that
\begin{itemize} 
\item[(i)] for all regular initial data $y_0\in T\Q$, the Cauchy
  problem for the system \eqref{non.lag} with the
  boundary conditions \eqref{stress.free} is locally well-posed;
\item[(ii)] let $y(t)$ be a non-impacting solution. Then its time of
  existence is infinite and it is forever non singular.
\item[(iii)] Any regular solution fulfilling
  $\sup_{\bx\in\B,t>0}|\zeta(\bx,t)|<\infty$ is precompact;

\item[(iv)] the angular momentum $\bL$ is conserved along the
  solutions; the Lie derivative of the energy \eqref{ham} is
  nonpositive and vanishes if and only if $\dot C=0$, where $C$ is the
  Cauchy Green tensor (cf. eq. \eqref{Cauchy}).
\end{itemize}
\end{hyp}

\begin{remark}
\label{galerkin}
If one thinks that the nonconservative equations we are studying
behave like parabolic equations, then typically the situation of
Assumption \ref{hyp:cauchy.problem} is fulfilled. Indeed, for positive
times the solution of a parabolic equation is contained in all Sobolev
spaces. Existence of global solutions has been proved for some equations of viscoelasticity (see \cite{KaShi92} and \cite{LLZ08}). Moreover, the property of having a compact attractor is known
for many classes of evolution equations involving dissipation (see
\cite{temam}).
\end{remark}


\begin{remark}
\label{gale2}
A situation in which Assumption \ref{hyp:cauchy.problem} is fulfilled
is that in which the space $\Q$ is finite dimensional. A typical
situation we have in mind is that in which the space $\Q$ is composed
by maps obtained by Galerkin cutoff from some of the maps belonging to the
original infinite dimensional configuration space. For example one
could decide to keep only a finite (arbitrarily large) number of
spherical harmonics of the maps describing the configuration.
\end{remark}

\subsection{The nondegeneracy assumption}\label{n.deg} In the
following (see subsect. \ref{sec:char})
we will prove that the nondissipating orbits are relative equilibria of
the Hamiltonian system obtained by Legendre transforming the
Lagrangian \eqref{lag}. We are now going to recall the notion of
relative equilibrium and to state the nondegeneracy condition we need.

 Define the momentum 
\begin{equation}
\label{momenta}
\pi:=\rho_0\dot \zeta\ ,
\end{equation}
then the Hamiltonian of the system coincides with the function $\H$
given by \eqref{ham}, where however
\begin{equation}
\label{k}
K:=\int_{\B}\frac{|\pi(\bx)|^2}{2\rho_0(\bx)}
d^3\bx \ ;
\end{equation}
and the Lagrange equations \eqref{lag.eq} are
equivalent to the Hamilton equations of \eqref{ham}.

The Hamiltonian is invariant under the action of the symmetry group
$SO(3)$ defined by
\begin{equation}
\label{sym.ham}
Rz\equiv R(\pi,\zeta):=(R\pi,R\zeta)\ ,
\end{equation}
and the total angular momentum $\bL$ (written in terms of positions
and momenta) is the corresponding conserved quantity.
Then one can use Marsden-Weinstein reduction procedure, that can be
summarized as follows.
\begin{itemize}
\item[(1)] Fix a value $\bL_0$ of $\bL$ and consider the manifold 
$$
\M_{\bL_0}:=\left\{z\equiv (\pi,\zeta)\ :\ \bL(z)=\bL_0
\right\}\ ; 
$$
\item[(2)] Consider the subgroup $\G_{\bL_0}\subset SO(3)$ leaving
  invariant $\M_{\bL_0}$, namely the group of the rotations around the
  axis $\bL_0$. Consider the quotient manifold
  $\M_{\bL_0}/\G_{\bL_0}$. Such a manifold has a natural symplectic
  structure and furthermore the Hamiltonian $\H$, as well as its
  Hamilton equations, pass to the quotient and define a Hamiltonian
  system on $\M_{\bL_0}/\G_{\bL_0}$.
\end{itemize}

We denote by $\H_{\bL_0}$ the Hamiltonian of such a reduced system. 

\begin{definition}
\label{red.eq}
The critical points of  $\H_{\bL_0}$ are called relative equilibria of
the Hamiltonian system $\H$, at angular momentum $\bL_0$.
\end{definition}

\begin{definition}
\label{nondeg}
A relative equilibrium is said to be topologically nondegenerate if it
is not an accumulation point of relative equilibria with the same
angular momentum.
\end{definition}

By abuse of notation a representative $z_e$ of the equivalence class of a
relative equilibrium is also called a relative equilibrium of $\H$.

\begin{remark}
\label{r.11}
It is well known (see e.g. \cite{abrahammarsden}) that $z_e$ is a
relative equilibrium if and only if the Hamiltonian vector field of
$\H$ at $z_e$ is tangent to the orbit of $SO(3)$ through $z_e$. 
\end{remark}
\begin{remark}
\label{r.12}
If $z_e$ is a relative equilibrium then the corresponding orbit $z(t)$
(under the flow of the Hamiltonian system $\H$) is formed by relative
equilibria. In the nondegenerate case there are no other relative
equilibria with the same angular momentum in a neighborhood of the
orbit $\cup_tz(t)$.
\end{remark}

Of course a relative equilibrium $z_e$ corresponding to a value ${\bf
  L}_0$ of $\bL$ is topologically nondegenerate if and only if the
same is true for the relative equilibrium $Rz_e$, where $R\in SO(3)$
is arbitrary.

\begin{definition}
\label{tndp}
A value $\ell\in\Re$ of the modulus of the angular momentum is said to
be nondegenerate if all the relative equilibria with angular momentum
${\bf L}$ satisfying $|{\bf L}|=\ell$ are topologically
nondegenerate relative equilibria of $\H$.
\end{definition}

\begin{remark}
\label{rem.non.deg}
In Section \ref{assu.deg} we will comment on this condition and
show that it is in general fulfilled. 
\end{remark}

\begin{teo}\label{main}
Let $y(t)$ be a solution of eq. \eqref{non.lag} with the boundary
condition \eqref{stress.free}, and let $\ell$ be the corresponding
value of the modulus of the angular momentum. Assume that $\ell$ is
nondegenerate, then, one of the following three (future) scenarios is
possible:
\begin{itemize}
\item[(i)] the trajectory of $\B$ is unbounded;
\item[(ii)] the solution impacts the planet;
\item[(iii)] the solution is asymptotic to a synchronous
  non-dissipating orbit, which is a relative equilibrium with angular
  momentum $\ell$.
\end{itemize}
\end{teo}

\section{Proof of theorem \ref{main}}
\label{sec.3}

\subsection{LaSalle's invariance principle}
\label{sec:lasalle}

In order to study the dynamics of the system, we make use of the
LaSalle's principle which is a refinement of the classical Lyapunov's
theorem. We now recall its statement and proof.

Let   $\Y$ be a Banach space and let $\U\subset\Y$ be open. Consider a
system of differential equations
\begin{equation}\label{eq.diff}
\dot y=f(y)\qquad y \in \U\ ,
\end{equation}
We denote by $\varphi$ the flow of \eqref{eq.diff}, which we assume to
be locally well defined. 

\begin{defi}
Let $\gamma$ be the orbit of \eqref{eq.diff} with initial condition
$y_0$. A point $\eta$ is said to be an $\omega$-limit point of
$\gamma$ if there exists a sequence of times $t_n\rightarrow +\infty$
such that
\begin{equation}
\lim_{n\rightarrow +\infty}\varphi^{t_n}(y_0)=\eta\ .
\end{equation}
\end{defi}

\begin{defi}
The $\omega$-limit set of an orbit $\gamma$ is defined as the union of
all the $\omega$-limit points of $\gamma$.
\end{defi}

\begin{definition}
\label{precomp.2}
A solution $y(t)\subset\U$, $t\in(0,\infty)$ is said to be precompact
if, for any increasing sequence $\{t_n\}\subset(0,\infty) $ there
exists a subsequence $\{t_{n_k}\} $ s.t. the limit
$\lim_{k\to+\infty}y(t_{n_k})$ exists and belongs to $\U$.
\end{definition}

\begin{remark}
\label{simply}
It is well known that the $\omega$-limit of a precompact orbit is a
connected set (see e.g. \cite{abrahammarsden}). 
\end{remark}

The classical version of LaSalle's principle may be stated as follows:

\begin{teo}\label{th:lasalle}
Suppose that $\H:\U\to\Re$ is a real-valued smooth function, such that
$\lie_f\H(y)\leq0$, $\forall y\in\U$, where $\lie_f$ is the Lie
derivative. Let $\I$ be the largest {\it invariant} set contained in
$\N:=\left\{y \in \U|\lie_f\H(y)=0\right\}$, then the $\omega$-limit
of every precompact orbit is a non-empty subset of $\I$.
\end{teo}

\begin{proof} Let $\gamma:=\left\{\varphi^t(y_0)|t>0\right\}$ be a precompact 
orbit, and let $\Gamma\subset \U$ be the $\omega$-limit of $\gamma$.
We prove now that $\Gamma$ is invariant. Indeed, let $\eta\in\Gamma$,
then there exists a sequence $t_n\to +\infty$ such that
$\varphi^{t_n}(y_0)\to\eta$. But we have
$$\varphi^t(\eta)=\varphi^t(\lim_{n\to +\infty}\varphi^{t_n}(y_0))=\lim_{n\to +\infty}\varphi^{t+t_n}(y_0)\in\Gamma\ .$$

We prove now that the $\omega$-limit is contained in $\N$. Let
$\eta_0\in\Gamma$. Then there exists a sequence $t_n\to +\infty$ such
that $\varphi^{t_n}(y_0)\to\eta_0$. Now,
let $$c:=\H(\eta_0)=\lim_{n\to +\infty}\H[\varphi^{t_n}(y_0)]\ .$$
Since $\H[\varphi^t(y_0)]$ is a time-nonincreasing function,
$\lim_{n\to +\infty}\H[\varphi^{t_n}(y_0)]=c$ is independent of the
subsequence $t_n$, and thus $\H=c$ on the whole $\Gamma$. By the
invariance of $\Gamma$ it follows that $\lie_f\H=0$ on $\Gamma$, and
therefore $\gamma\subset \I$.
\end{proof}

\subsection{Non-dissipating orbits}
\label{sec:char}

Consider the non dissipating manifold defined by
\begin{equation}
\label{non.dissi}
\N:=\left\{y\in\Q\oplus\Q\ :\ \dot \H(y)=0\right\}\ ,
\end{equation}
where, for short we denoted by $\dot \H(y)$ the Lie derivative of $\H$
along the vector field corresponding to the equations
\eqref{non.lag}. 
In this section we prove that the subset $\I\subset\N$ invariant under
the dynamics is formed by relative equilibria of the Hamiltonian
system \eqref{ham}.

\begin{remark}
\label{nd}
On $\N$ the Lagrange equations \eqref{lag.eq} coincide with the non
conservative equations \eqref{non.lag}.
\end{remark}

First we prove that the body is rigid along any orbit in $\I$ (we
think that this should be well known, but we were not able to find a
reference).
\begin{lemma}\label{lemma.distance}
Let $y\in C^2((0,+\infty),\Q\oplus\Q)$ be a solution of \eqref{non.lag}
s.t. $ y(t)\equiv(\zeta(t),\dot\zeta(t))\in\N$ $\forall
t\in(0,+\infty)$, then, for all $\bx,\by\in\B$, one has
\begin{equation}\label{dist.preserved}
\frac{d}{dt}|\zeta(\bx)-\zeta(\by)|=0\ .
\end{equation}
\end{lemma}

\begin{proof}
Fix two arbitrary points $\bx,\by\in\B$ and consider a path $\gamma\subset\B$ connecting $\bx$ to $\by$. Let $s$ be the arclength parameter, so that $|\frac{d\gamma(s)}{ds}|=1$ and in the reference configuration the path $\gamma$ has length
\begin{equation*}
\length{\gamma}=\int_{\gamma}ds\ .
\end{equation*}
The length of the deformed path $\zeta(\gamma)$ is expressed in terms of the Cauchy Green tensor $C$ by (see e.g. \cite{slaughter})
\begin{equation}
\length{\zeta(\gamma)}=\int_{\gamma}\left[\left(\frac{d\gamma(s)}{ds}\right)^TC\left(\frac{d\gamma(s)}{ds}\right)\right]^{\frac{1}{2}}ds\ .
\end{equation}
Therefore, since $\dot C\equiv 0$, we have
$\frac{d}{dt}[\length{\zeta(\gamma)}]\equiv 0$.  Now, take two
arbitrary times $t_0,t_1$ and let $\zeta_0,\zeta_1$ be the
corresponding body configurations. We have that
$\zeta_1\circ(\zeta_0)^{-1}:\zeta_0(\B)\to\zeta_1(\B)$ is a
length-preserving map between the  $\zeta_0(\B)\subset\Re^3$ and
$\zeta_1(\B)\subset\Re^3$, both equipped with the restriction of the Euclidean
metric on $\Re^3$. Moreover,
$\zeta_1\circ(\zeta_0)^{-1}$ is a diffeomorphism. 

Using the fact that the segments minimize the distance it is easy to
conclude the proof of the lemma (some care is needed in order to take
care of the fact that $\zeta_1(\B)$ could fail to be convex).
\end{proof}

\begin{corollary}
\label{rigid}
Let $y=(\zeta,\dot\zeta)$ be as in the statement of Lemma
\ref{lemma.distance} then there exist $\xi\in\Q$, $R\in
C^2((0,+\infty),SO(3)) $ and $\bY\in C^2((0,+\infty),\Re^3)$ s.t.
\begin{equation}
\label{rigid.motion}
\zeta(\bx,t)=R(t)(\xi(\bx)+\bY(t))\ .
\end{equation}
\end{corollary}

The reference frame with origin $\bY(t)$ and coordinate axes
$R(t)\be_i$ is usually called {\it comoving frame}. In this frame all
the points of the satellite are at rest along the orbit $y(t)$. In
particular $-\bY(t)$ is the position of the planet $M$ in the comoving
frame.

\begin{lemma}\label{der.terze}
Let $y\in C^2((0,+\infty),\Q\oplus\Q)$ be a solution of \eqref{non.lag}
s.t. $ y(t)\equiv(\zeta(t),\dot\zeta(t))\in\N$ $\forall
t\in(0,+\infty)$, then the quantity $\bY(t)$ in \eqref{rigid.motion}
evolves in such a way that
\begin{equation}\label{d.3}
\forall
i,j,k,\ \frac{\partial^3V_g}{\partial\chi^i\partial\chi^j\partial\chi^k}(\bY(t))=0
\end{equation}
is independent of time.
\end{lemma}

\begin{proof} Inserting the expression \eqref{rigid.motion} in the
  Lagrange equations \eqref{lag.eq} and exploiting the rotational
  invariance of the r.h.s. (cf. Remark \ref{rot.inv1}) one gets the
  following equation for $\xi$, $R$ and $\bY$:
\begin{align}
\label{com.eq}
&\left[\ddot \bY+(\dot{\hat
    \omega}+\hat\omega\hat\omega)\bY+2\hat\omega\dot\bY + (\dot{\hat
    \omega}+\hat\omega\hat\omega)\xi\right] \\ \nonumber &=
-\frac{\partial V_g}{\partial \chi}(\xi+\bY)-\frac{\partial
  V_{sg}^{\zeta}}{\partial \chi}(\xi)+
\frac{1}{\rho_0}\frac{\partial}{\partial x^a}\frac{\partial
  W}{\partial(\partial \zeta/\partial x^a)} (D\xi) \ ,
\end{align}
where as usual $\hat \omega:=R^T\dot R$. Denote for short
\begin{equation}
\label{def.L}
L:=(\dot{\hat \omega}+\hat\omega\hat\omega) 
\end{equation}
and take the time
derivative of \eqref{com.eq}. Taking into account that $\xi$ does not
depend on time one gets
\begin{equation}
\label{com.eq.2}
\left\{\frac{d}{dt} \left[\ddot \bY+L\bY+2\hat\omega\dot\bY \right]
+ \dot L \xi\right\} =\frac{d}{dt} \left[ -\frac{\partial
  V_g}{\partial \chi}(\xi+\bY) \right] \ .
\end{equation}
Take now the derivative of such a quantity with respect to $x^a$, one
gets the componentwise equation
\begin{equation}
\label{com.eq.3}
\left(\sum_k \dot L^i_k \frac{\partial \xi^k}{\partial x^a}\right) =\sum_k
\frac{d}{dt} \left[ -\frac{\partial^2 V_g}{\partial
    \chi^i\partial\chi^k}(\xi+\bY) \frac{\partial \xi^k}{\partial
    x^a}\right] \ ,
\end{equation}
or using the invertibility of the matrix $\frac{\partial \xi^k}{\partial
    x^a}$
\begin{equation}\label{d.7}
\dot L^i_{k}=-\frac{d}{dt}\left[\frac{\partial^2V_g}{\partial\chi^i\partial\chi^k}(\xi(\bx)+\bY(t))\right]\ .
\end{equation}
This equation shows in particular that the r.h.s. is independent of
$\bx$. Due to the invertibility of $\xi$ and to the analyticity of
$V_g$ this means that the function of $\chi$
\begin{equation}
\label{analitic}
-\frac{d}{dt}\left[\frac{\partial^2V_g}{\partial\chi^i\partial\chi^k}(\chi+\bY(t))\right]\
\end{equation}
is actually independent of $\chi$. Thus taking the derivative with
respect to $\chi^j$ and evaluating at $\chi=0$ one gets the thesis.
\end{proof}

\begin{lemma}\label{Mfixed}
Let $y\in C^2((0,+\infty),\Q\oplus\Q)$ be a solution of \eqref{non.lag}
s.t. $ y(t)\equiv(\zeta(t),\dot\zeta(t))\in\N$ $\forall
t\in(0,+\infty)$, then the quantity $\bY(t)$ in \eqref{rigid.motion}
is actually independent of time.
\end{lemma}

\begin{proof}
We write down explicitly \eqref{d.3}. We denote
$\left(Y^1(t),Y^2(t),Y^3(t)\right)=\bY(t)$ and
$(Y^1)^2+(Y^2)^2+(Y^3)^2=r^2$. One has
\begin{equation}\label{diff3.1}
\frac{\partial^3V_g}{\partial (\chi^1)^3}(\bY)=\frac{3kM Y^1(5(Y^1)^2-3r^2)}{r^7}
\end{equation}
\begin{equation}\label{diff3.2}
\frac{\partial^3V_g}{\partial (\chi^1)^2\partial \chi^2}(\bY)=\frac{3kMY^2(5(Y^1)^2-r^2)}{r^7}
\end{equation}
\begin{equation}\label{diff3.3}
\frac{\partial^3V_g}{\partial (\chi^1)^2\partial \chi^3}(\bY)=\frac{3kMY^3(5(Y^1)^2-r^2)}{r^7}\ .
\end{equation}
Choose now the comoving frame in such a way that $Y^2=Y^3=0$ and $Y^1=r$ at
$t=0$ (which is possible up to redefinition of $R$ and $\xi$). By \eqref{diff3.2} and \eqref{diff3.3}, for any $t$ one must
have $$Y^2(5(Y^1)^2-r^2)=0$$ $$Y^3(5((Y^1)^2-r^2)=0$$ which by continuity implies
$Y^2(t)\equiv Y^3(t)\equiv 0$ and $Y^1(t)\equiv r(t)$. Finally, substituting
in \eqref{diff3.1}, one has $$\frac{1}{(Y^1(t))^4}=\frac{1}{(Y^1_0)^4}\ ,$$
whose only solution is $Y^1(t)\equiv Y^1_0$.
\end{proof}

Remark now that, given a non dissipating solution, one can associate to
it a shape of the body described by the function $\xi(\bx)+\bY$, and
the shape evolves by a rigid motion about the fixed point $M$. 
Introduce the angular velocity which is defined as usual as the vector
$\omega$ s.t. the two operators
$$
\omega\times\cdot=\hat \omega\ 
$$
coincide. Then the velocity of the motion, in the comoving frame, is
given by $\omega\times (\xi+\bY)$.

We have now that, for a non dissipating solution, $\omega$ does not depend on
time.
\begin{lemma}
\label{omega}
The angular velocity $\omega$ of a nondissipating solution $y$ is
independent of time.
\end{lemma}
\begin{proof}
Consider again equation \eqref{d.7}. Since we now know that $\bY$ is
independent of time it follows that the operator
$L$ (cf. eq. \eqref{def.L}) fulfills $\dot L=0$. This means that, for
any vector $\chi$ one has
\begin{equation}
\label{ang.v}
\frac{d}{dt}\left[\dot \omega\times \chi+\omega\times(\omega\times
  \chi)\right] = 0\ .
\end{equation}
To exploit such an equation take $\chi=\be_i$ and project the square
bracket on $\be_i$. Using standard vector identities this implies
$$
\frac{d}{dt} \left|\omega\times \be_i\right|^2=0\ \quad \forall i\ .
$$
A trivial computation shows that this implies $\dot \omega=0$. 
\end{proof}

\begin{corollary}
\label{rel.eq}
Let $y(t)$ be a nondissipating solution as above, then it is the orbit
of a relative equilibrium of the system \eqref{ham}.
\end{corollary}
\begin{proof}
We have proved that along a non dissipating solution
$\zeta(t)=R(t)\zeta_0$ with a suitable configuration $\zeta_0$ and a
rotation matrix $R(t)$ that we can choose in such a way that
$R(0)=I$. It follows $\dot \zeta(t)=R(t)[\omega\times
  \xi]=R\dot\zeta(0)$. Passing to the phase space one gets that along
such an orbit $\pi(t)=\rho_0\dot\zeta(t)=R(t)\pi(0)$. This shows that
the solution is actually an orbit of the symmetry group, and this is a
characterization of being a relative equilibrium.
\end{proof}

Thus we have that the manifold $\I$ is the union of the trajectories
of all the possible relative equilibria of the system.

\subsection{End of the proof}\label{end}

 Applying La Salle principle to our system, with $\U$ defined to be
 the set of regular configurations, we get the following Lemma.

\begin{lemma}\label{three.outcomes.0}
For any solution to \eqref{non.lag} with the boundary condition
\eqref{stress.free}, one of the following three (future) scenarios
occur:
\begin{itemize}
\item[(i)] the trajectory is unbounded;
\item[(ii)] the solution impacts the planet;
\item[(iii)] the solution is asymptotic to the non-dissipating
  invariant manifold $\I$. 
\end{itemize}
\end{lemma}

\noindent{\it End of the proof of Theorem \ref{main}.} Let $\bL_0$ be
the initial value of the angular momentum, then a regular bounded
solution is asymptotic to 
\begin{equation}
\label{lim}
\I\cap \left\{(\zeta,\dot\zeta)\in T\Q\ :\ \bL(\zeta,\dot\zeta
)=\bL_0\right\}\ ,
\end{equation}
but, by the nondegeneracy assumption the set \eqref{lim} is formed by
orbits which are isolated in the invariant manifold $
\left\{(\zeta,\dot\zeta)\in T\Q\ :\ \bL(\zeta,\dot\zeta
)=\bL_0\right\}$. Thus, by Remark \ref{simply} the $\omega$-limit set
of an orbit is a single orbit in the set \eqref{lim}, i.e. a
synchronous orbit. \qed 

\section{On the nondegeneracy assumption}\label{assu.deg}

In this section we are going to prove that if the restoring elastic
forces described by the potential \eqref{ue} are strong enough, then
the nondegeneracy of a relative equilibrium for the system \eqref{lag}
is equivalent to the nondegeneracy of the relative equilibrium for a
rigid body having the shape given by the asymptotic configuration of
the satellite. 

For simplicity, in this section we limit the discussion at the formal
level, namely we forget all the difficulties related to the existence
of unbounded operators. All what follows is rigorous if $\Q$ is finite
dimensional. It can also be made rigorous in the case of PDEs by
detailing most of the assumptions, following the ideas of \cite{AP}
and exploiting the ellipticity properties of the elasticity tensor
(see \cite{marsdenhughes}), however this is outside the aims of the
present paper.

First of all, having fixed a configuration $\bar \zeta$ we introduce
the dynamical system describing the evolution of a rigid satellite
with shape $\bar\zeta$. The configuration space $SO(3)\times
\Re^3\ni(R,\chi)$ and the dynamics is obtained from the Lagrangian
obtained by restricting the
Lagrangian \eqref{lag} to the set of motions of the form
\begin{equation}
\label{rigid.1}
\zeta(t)=R(t)[\bar \zeta+\chi(t)]\ .
\end{equation}
Denote by $\L_{\bar\zeta}$ such a Lagrangian.  One also has a
corresponding Hamiltonian system, which is deduced from $\eqref{lag}$
in the standard way and whose Hamiltonian coincides with the
restriction of the Hamiltonian \eqref{ham} to the phase space of the
rigid body. Denote by $H_{\bar\zeta}$ the Hamiltonian of the rigid
body. Such a Hamiltonian is invariant under rotations so one can pass
to the reduced system and to introduce again the relative equilibria
and define the nondegenerate relative equilibria for such a system and
a nondegenerate value of the modulus of the angular momentum. 

In order to make a connection between the nondegeneracy of the
relative equilibria for rigid motions and the nondegeneracy of the
elastic motions we need to specify an assumption on the elastic
potential energy $U_e$. Essentially we are going to assume that the
elastic potential has a very steep, isolated (up to the symmetries)
minimum at some shape.

First remark that $U_e$ is invariant under the action 
\begin{equation}
\label{group.2}
(SO(3)\times\Re^3)\times \Q\ni ((R,\chi),\zeta)\mapsto
R[\zeta+\chi]\in\Q\ ,
\end{equation}
 and, given a point $\zeta$, consider the group orbit
$\G_\zeta:=(SO(3)\times\Re^3)\zeta\subset \Q$, so, if $\bar \zeta$ is a
 critical point of $U_e$, then all the orbit $\G_{\bar\zeta}$ is
 critical for $U_e$.

\begin{hyp}\label{suUe}
One has $U_e=\frac{1}{\epsilon}\tilde U_e$, and $\tilde U_e$ is a
smooth function invariant under the group action \eqref{group.2} with
the further property that the set of its critical points is formed by
finitely many orbits $\G_{\zeta^{(i)}}$ and each critical point is
nondegenerate in the direction transversal to the group
orbit\footnote{Transversal nondegeneracy means that the restriction of
  $\tilde U_e$ to any hyperplane transversal to the group orbit has a
  differential which is an isomorphism.}.
\end{hyp}

Under this assumption we have the following 
\begin{proposition}
\label{propo.nondeg}
Fix a value $\bL_0$ of the angular momentum, and assume $\epsilon$ is
small enough; let $\zeta_e^{\epsilon}$ be a relative equilibrium of
the Hamiltonian system $H$ with angular momentum $\bL_0$. If $(R,0)\in
SO(3)\times \Re^3$ is a nondegenerate relative equilibrium for the
rigid system with Hamiltonian $H_{\zeta_e^{\epsilon}}$ then
$\zeta_e^{\epsilon}$ is a nondegenerate relative equilibrium for the
elastic Hamiltonian system with Hamiltonian $H$.
\end{proposition}
 
\begin{proof} The proof is based on ideas from Lyapunov-Schmidt
  decomposition (see e.g. \cite{AP}).
First of all it is useful to introduce suitable coordinates in $\Q$ in a
neighborhood of $\zeta_e$ (following the ideas of \cite{bambusihaus}). They
are constructed as follows: let $\Sigma\subset \Q$ be a codimension 6 affine
subspace transversal to the group orbit $\G_{\zeta_e}$, then a suitable
set of coordinates about $\zeta_e$ is locally obtained by the map
\begin{equation}
\label{coord.1}
\Sigma\times  SO(3)\times \Re^3\ni (\xi,R,\chi)\mapsto R(\xi+\chi)\ .
\end{equation}

We now recall (and adapt to the present situation) some results of
\cite{SPM91}. Using the method of Lagrange multipliers one immediately
sees that the relative equilibria can be obtained  by finding the
critical points of $H-\omega\cdot(\bL-\bL_0)$ under the
condition  
\begin{equation}
\label{l=l0}
\bL=\bL_0\ .
\end{equation} Here $\omega$ is the Lagrange multiplier. In
\cite{SPM91} it was shown that this is equivalent to finding the critical
points of the ``augmented Hamiltonian'' $H_{\bL_0}$ defined by 
\begin{align}
\label{sys.cri}
H_{\bL_0}:=K_{\bL_0}+U_{\bL_0}\ ,
\\
\label{sys.cri1}
K_{\bL_0}:=\frac{1}{2}\int_{\B}\frac{\left|\pi-\rho_0(\omega\times
  \zeta)\right|^2}{\rho_0}d^3\bx\ ,
\\
\label{sys.cri2}
 U_{\bL_0}:=U-\frac{1}{2}\int_{\B}\rho_0\left|\omega\times
 \zeta\right|^2 d^3\bx
\end{align}
again under the condition \eqref{l=l0} (actually in the present case
this is a straightforward computation). As pointed out in \cite{SPM91}
the interest of this formulation is that the equations for the
critical points of $H_{\bL_0}$ take the form
\begin{align}
\label{cri.3}
\pi_e=\rho_0\omega\times \zeta_e\ ,
\\
\label{cri.4}
\nabla U(\zeta_e)+\rho_0\omega\times (\omega\times\zeta_e)=0\ .
\end{align}
In particular the second equation is independent of $\pi$. This allows
to study separately \eqref{cri.4}. To this end we use the set of
coordinates \eqref{coord.1}. However one has to pay attention to the
fact that in general the section $\Sigma$ is not orthogonal to the
group orbit, so first we rewrite \eqref{cri.4} in
the (original) dual form:
\begin{equation}
\label{cri.5}
d U(\zeta_e) h + \left\langle \rho_0\omega\times
(\omega\times\zeta_e);h  \right\rangle_{L^2}=0\ ,\quad \forall h\in\Q\ ,
\end{equation}
which in terms of the coordinates \eqref{coord.1} takes the form (at $R=I$)
\begin{align}
\label{cri.6}
d_\chi U_g(\xi_e+\chi_e) h_{\chi} &+ \left\langle \rho_0\omega\times
(\omega\times(\xi_e+\chi_e);h_\chi \right\rangle_{L^2}=0\ ,\quad \forall
h_\chi\in\Re^3\ , 
\\ 
\label{cri.7}
d_\xi
[U_g(\xi_e+\chi_e)&+U_{sg}(\xi_e+\chi_e)+\frac{1}{\epsilon}\tilde U_e(\xi_e+\chi_e)]
h_\xi 
\\
\nonumber
&+ \left\langle \rho_0\omega\times (\omega\times(\xi_e+\chi_e));h_\xi
\right\rangle_{L^2}=0\ ,\quad \forall h_\xi\in T_{\zeta_e}\Sigma\ .
\end{align}
which have to be solved together with 
\begin{equation}
\label{cri.8}
\pi_e=\rho_0\omega\times (\xi_e+\chi_e)\ ,
\end{equation}
and the condition \eqref{l=l0}. In particular the system
\eqref{cri.6}, \eqref{cri.8}, \eqref{l=l0} is identical to the system
for the reduced equilibrium of the rigid body with shape $\zeta_e$, so
by Assumption \ref{suUe} it determines uniquely (up to a finite choice) $\chi_e$ and $\omega$ (and
$\pi_e$). We analyze now \eqref{cri.7}. Of course it is a perturbation
of $d_{\xi}\tilde U_e=0$, whose critical points are all nondegenerate
(as functions of $\xi$ they are nondegenerate in the standard sense),
so, by the implicit function theorem a solution $\xi_e$ of
\eqref{cri.7} must be close to a critical point of $\tilde U_e$ and is
also nondegenerate. This concludes the proof of the proposition.
\end{proof}

So we have reduced the problem of checking nondegeneracy to the
problem of checking nondegeneracy of the relative equilibria for the
motions of rigid bodies. This problem has been studied for example in
\cite{teixidoroman} who obtained a complete characterization of the
relative equilibria of a triaxial rigid body, provided the
gravitational potential is approximated by its quadrupole expansion. In
\cite{teixidoroman} the author obtained that, provided the distance
$\chi_e$ of the center of mass from the planet is large enough, there
are exactly 24 families of stationary points of the reduced
system. These stationary points are such that the principal axes of
inertia are one pointing to $M$ and a second one in the plane
orthogonal to the plane of motion. The number 24 appears as the number
of possible choices of the orientations of the body with prescribed
principal axes of inertia.

Remark that in particular it turns out that such critical points are
nondegenerate. Furthermore we expect that, if $\chi_e$ is large
enough, then it should be possible to use the implicit function
theorem to prove that the critical points are nondegenerate also for
the system in which the gravitational potential is not subjected to
any approximation.

\begin{remark}
\label{bamhaus}
In \cite{bambusihaus} a problem is studied in which Assumption
\ref{suUe} is violated due to the fact that the satellite is assumed to
have a spherically invariant reference equilibrium configuration. In
that case it was shown that actually the orbit is not asymptotic to a
single synchronous orbit; nevertheless, one can still say that also in
that case the orbit is asymptotic to a synchronous resonance, but the
principal axes of inertia could rotate in the satellite.  
\end{remark}


\begin{thebibliography}{10}

\bibitem{abrahammarsden}
R.~{Abraham} and J.~E. {Marsden}.
\newblock {\em {Foundations of mechanics} (Second edition)}.
\newblock Benjamin/Cummings Publishing, Reading (Massachusetts), 1978.

\bibitem{alexander}
M.~E. {Alexander}.
\newblock {The Weak Friction Approximation and Tidal Evolution in Close Binary
  Systems}.
\newblock {\em Astrophysics and Space Science}, 23:459--510, August 1973.

\bibitem{AP} A.~Ambrosetti and G.~Prodi, \newblock {\em {A primer of
    nonlinear analysis}}. {Cambridge University Press}, Cambridge, {1993}.


\bibitem{arnold}
V.~I. {Arnold}.
\newblock {\em {Mathematical Methods of Classical Mechanics}}.
\newblock Springer, Berlin, 1978.

\bibitem{ball}
J.~M. {Ball}.
\newblock {Some open problems in elasticity}.
\newblock In:{\em Geometry, Mechanics, and Dynamics.} Springer-Verlag New York, New York, pp. 459--510, 2002.

\bibitem{bambusihaus}
D.~{Bambusi} and E.~{Haus}.
\newblock {Asymptotic stability of synchronous orbits for a gravitating viscoelastic sphere}.
\newblock {\em Celestial Mechanics and Dynamical Astronomy}, Volume 114, Issue 3, 255--277, 2012.

\bibitem{darwin1}
G.~H. {Darwin}.
\newblock {On the Precession of a Viscous Spheroid, and on the Remote History
  of the Earth}.
\newblock {\em Philosophical Transactions of the Royal Society of London Series I},
  170:447--538, 1879.

\bibitem{darwin2}
G.~H. {Darwin}.
\newblock {On the Secular Changes in the Elements of the Orbit of a Satellite
  Revolving about a Tidally Distorted Planet}.
\newblock {\em Philosophical Transactions of the Royal Society of London Series I},
  171:713--891, 1880.
  
\bibitem{efroimsky1}
M.~{Efroimsky} and J.~G. {Williams}.
\newblock {Tidal torques. A critical review of some techniques}.
\newblock {\em Celestial Mechanics and Dynamical Astronomy}, 104:257--289, 2009.

\bibitem{efroimsky2}
M.~{Efroimsky}.
\newblock {Bodily tides near spin-orbit resonances}.
\newblock {\em Celestial Mechanics and Dynamical Astronomy}, 112:283--330, 2012.  
  
\bibitem{ferrazmello}
S.~{Ferraz-Mello}, A.~{Rodr{\'{\i}}guez}, and H.~{Hussmann}.
\newblock {Tidal friction in close-in satellites and exoplanets: The Darwin
  theory re-visited}.
\newblock {\em Celestial Mechanics and Dynamical Astronomy}, 101:171--201, May
  2008.  
  
\bibitem{goldreich}
P.~{Goldreich}.
\newblock {Final spin states of planets and satellites}.
\newblock {\em Astronomical Journal}, 71:1--7, February 1966.

\bibitem{kaula}
W.~M. {Kaula}.
\newblock {Tidal Dissipation by Solid Friction and the Resulting Orbital
  Evolution}.
\newblock {\em Reviews of Geophysics and Space Physics}, 2:661--685, 1964.

\bibitem{KaShi92}
S.~{Kawashima} and Y.~{Shibata}.
\newblock {Global Existence and Exponential Stability of Small Solutions to Nonlinear Viscoelasticity}.
\newblock {\em Communications in Mathematical Physics}, 148:189--208, 1992.

\bibitem{lasalle}
J.~P. LaSalle.
\newblock {\em {The stability of dynamical systems}}.
\newblock Society for Industrial and Applied Mathematics, Philadelphia, Pa.,
  1976.
\newblock With an appendix: ``Limiting equations and stability of nonautonomous
  ordinary differential equations'' by Z. Artstein, Regional Conference Series
  in Applied Mathematics.
  
\bibitem{LLZ08}
Z.~{Lei}, C.~{Liu}, and Y.~{Zhou}.
\newblock {Global Solutions for Incompressible Viscoelastic Fluids}.
\newblock {\em Archive for Rational Mechanics and Analysis}, 188:371--398, 2008.
  
\bibitem{macdonald}
G.~J.~F. {MacDonald}.
\newblock {Tidal Friction}.
\newblock {\em Reviews of Geophysics and Space Physics}, 2:467--541, 1964.

\bibitem{marsdenhughes}
J.~E. Marsden and Thomas J.~R. Hughes.
\newblock {\em {Mathematical foundations of elasticity}}.
\newblock Dover Publications Inc., New York, 1994.
\newblock Corrected reprint of the 1983 original.

\bibitem{peale}
S.~J. {Peale}, P.~{Cassen}, and R.~T. {Reynolds}.
\newblock {Tidal dissipation, orbital evolution, and the nature of Saturn's
  inner satellites}.
\newblock {\em Icarus}, 43:65--72, July 1980.

\bibitem{SPM91}
J.~C. {Simo}, T.~A. {Posbergh}, and J.~E. {Marsden}.
\newblock {Stability of Relative Equilibria. Part II: Application to Nonlinear Elasticity}.
\newblock {\em Archive for Rational Mechanics and Analysis}, 115:61--100, 1991.

\bibitem{slaughter}
W.~S. Slaughter.
\newblock {\em {The linearized theory of elasticity}}.
\newblock Birkhauser, 2002.

\bibitem{teixidoroman}
M.~{Teixidó Román}.
\newblock {Hamiltonian methods in stability and bifurcation problems for artificial satellite dynamics}.
\newblock {Master's degree Thesis}, \url{http://upcommons.upc.edu/pfc/bitstream/2099.1/14225/1/memoria-8.pdf}, 2012.

\bibitem{temam}
R.~{Temam}.
\newblock {\em {Infinite-Dimensional Dynamical Systems in Mechanics and Physics}}.
\newblock Springer-Verlag, New York, 1997.


\end{thebibliography}

\end{document}